\documentclass[10pt]{article} 
\usepackage[preprint]{rlc}

\usepackage{algorithm}
\usepackage{algpseudocode}
\usepackage{amssymb}          
\usepackage{mathtools}        
\usepackage{mathrsfs}           
\usepackage{graphicx}           
\usepackage{subcaption}         
\usepackage[space]{grffile}   
\usepackage{url}             
\usepackage{soul}       
\usepackage{amsthm}
\usepackage{amsmath}
\usepackage{comment}
\usepackage{tikz}
\usepackage{booktabs}
\usepackage{makecell}
\usepackage{authblk}
\usetikzlibrary{shapes.geometric, arrows, calc, positioning}

\newtheorem{assumptions}{Assumptions}
\newtheorem{definition}{Definition}
\newtheorem{theorem}{Theorem}
\newtheorem{proposition}{Proposition}

\DeclareMathOperator*{\argmin}{argmin}

\newcommand*{\defeq}{\stackrel{\text{def}}{=}}

\definecolor{rougeINRIA}{RGB}{230, 51, 18}
\definecolor{grisbleuINRIA}{RGB}{56,66,87}
\definecolor{orangeINRIA}{RGB}{240, 126, 38}
\definecolor{lilasINRIA}{RGB}{155,0,79}
\definecolor{bleuINRIA}{RGB}{20,136,202}
\definecolor{vertINRIA}{RGB}{149,193,31}
\definecolor{jauneINRIA}{RGB}{255,205,28}
\definecolor{mauveINRIA}{RGB}{101,97,169}
\definecolor{bleuclairINRIA}{RGB}{137,204,202}
\definecolor{vertclairINRIA}{RGB}{199,214,79}

\title{Learning in Conjectural Stackelberg Games}

\author[1,*]{Francesco Morri}
\author[1]{Hélène Le Cadre}
\author[1]{Luce Brotcorne}

\affil[1]{Univ. Lille, CNRS, Inria, Centrale Lille, UMR 9189 CRIStAL, F-59000 Lille, France}
\affil[*]{Corresponding author: francesco.morri@inria.fr}

\begin{document}

\maketitle

\begin{abstract}
We extend the formalism of Conjectural Variations games to Stackelberg games involving multiple leaders and a single follower.
To solve these nonconvex games, a common assumption is that the leaders compute their strategies having perfect knowledge of the follower's best response. However, in practice, the leaders may have little to no knowledge about the other players' reactions. To deal with this lack of knowledge, we assume that each leader can form conjectures about the other players' best responses, and update its strategy relying on these conjectures. Our contributions are twofold: (i) On the theoretical side, we introduce the concept of Conjectural Stackelberg Equilibrium -- keeping our formalism conjecture agnostic -- with Stackelberg Equilibrium being a refinement of it. (ii) On the algorithmic side, we introduce a two-stage algorithm with guarantees of convergence, which allows the leaders to first learn conjectures on a training data set, and then update their strategies. Theoretical results are illustrated numerically. 
\end{abstract}

\section{Introduction}
Game theory has been used in many fields to model and understand the outcomes of situations involving strategic agents in competition \cite{baar_dynamic_1998}.
In this paper, we focus on bridging the gap between Stackelberg games, involving hierarchical decision making, and Conjectural Variations (CVs) games by introducing Conjectural variations Stackelberg Equilibrium (CSE) as a broad solution concept that encompasses Stackelberg Equilibrium (SE) as a refinement. 

\textbf{Conjectures in Games.} CV games are characterized by \textit{conjectures}, which are functions used by players to model their adversaries' reactions to their own action \cite{figuieres_theory_2004}. Conjectures are especially relevant in a context of bounded rationality \cite{simon_behavioral_1955}, where each player chooses its action based on its subjective perception of the other players' behaviors. 
CV games have been introduced in the context of duopolies \cite{bresnahan_duopoly_1981,friedman_bounded_2002,perry_oligopoly_1982}, considering classical Bertrand and Cournot equilibrium models. A new stream has appeared with the works of \cite{calderone_consistent_2023,chasnov_opponent_2019} who rely on variational analysis and dynamic system theory to provide generic interpretations of conjectural games.
CV games have not been extensively applied in real-world problems, though some applications can be found in the context of electricity markets \cite{diaz_electricity_2010,chen_conjectural-variations_2021}.

\textbf{Stackelberg Games.} We consider non-cooperative games involving a finite set of leaders at the upper level which compute their strategies based on their anticipations of the best response of the follower, which reacts rationally to the leaders' strategies. 
Some cornerstone theoretical contributions have emerged on the analysis of multi-leader single follower (MLSF) Stackelberg games \cite{sherali_multiple_1984,leyffer_solving_2010}, but algorithmic contributions in this area remain scarce. Applications of MLSF Stackelberg games can be found in jamming problems \cite{zhang_multi-leader_2018}, strategic bidding in deregulated markets \cite{morri_learning_2024}, security and privacy games \cite{gan_stackelberg_2018}. 

\textbf{Linking Conjectural and Stackelberg Games.} Most approaches in the literature consider leaders which have the ability to anticipate the best response of the follower. In reality, leaders rely on conjectures to compute their adversaries' best responses. However, the use of conjectures in Stackelberg games has not been extensively studied. Close works are related to dynamic conjectural games for duopolies \cite{figuieres_theory_2004} and reverse Stackelberg games \cite{groot_reverse_2012} where the leaders announce their strategies as functions of the follower's action. CV Equilibrium (CVE) has been introduced in \cite{olsder_phenomena_2009} as a relevant solution concept to cope with incomplete information and implicit cooperation, as well as a shorthand for dynamic interactions. In \cite{rubinstein_rationalizable_1994}, specific focus is put on the design of conjectures leading to a Nash Equilibrium. Conjectures can also be linked to models of bounded rationality, in place of prospect theory \cite{kahneman_prospect_1979,fochesato2025}, 
which requires to use complex subdifferential approaches to solve nonconvex nonsmooth games.

\textbf{Learning in Games with Hierarchical Structure.} There are two main approaches: either applying learning techniques to games, or using game theory to develop new learning algorithm. In our paper, we focus on the first class, leaving the possibility of future extensions in the second class. 
Learning in Stackelberg games is considered as a challenge. In this direction, \cite{fiez_implicit_2020} makes a significant contribution by developing a gradient-based learning rule for the leader, while the follower employs a best-response strategy. In the same vein, many algorithms involving a double-loop structure are proposed in the literature to compute stationary solutions \cite{grontas}. Single-loop stochastic algorithm has recently been proposed to tackle bilevel optimization problems \cite{hong}, allowing fast convergence rates. On the conjectural variations game side, \cite{wellman_conjectural_1998} propose learning dynamic rules for two player conjectural games. A similar idea is extended by \cite{chasnov_opponent_2019}, considering gradient-based learning methods for anticipating the adversary's reaction in general conjectural games. 

\textbf{Contributions.} We introduce CSE as a new solution concept to analyze Stackelberg games involving conjectures from the leaders on their adversaries' best responses. The distance between CV and Stackelberg game outcomes is upper bounded, and conditions on conjectures to reach a SE are identified. Then, we propose a two-stage algorithm to learn a CSE and provide convergence guarantees. Finally, the theoretical results are illustrated numerically.

\textbf{Structure.} In Sec.~\ref{sec:conj_stack} we introduce the main concepts of CV Games and MLSF Stackelberg Games, to then provide a definition for the novel class of Conjectural Stackelberg games and its equilibria. We then analyses CS Games in Sec.~\ref{sec:analysis}, relating them to standard Stackelberg Games. In Sec.~\ref{sec:algo} we introduce the \texttt{COSTAL} algorithm, with its proof of convergence, and in Sec.~\ref{sec:experiments} we provide numerical results obtained on multi player games.

\subsection{Notations}
We consider a set $\mathcal{N}\defeq \{1,\hdots,N\}$ of $N$ agents, each with strategy $x_i \in \mathbb{R}^{m_i},\,m_i \in \mathbb{N}^\star$, with $m\defeq \sum_i m_i$, and objective function $f_i : \mathbb{R}^{m} \rightarrow \mathbb{R}$. We define $\mathcal{X}_i \subseteq \mathbb{R}^{m_i}$ as the feasibility set of player $i$. We introduce $\mathcal{X}_{-i} \defeq \prod_{j\neq i}\mathcal{X}_j$ as the product of the feasibility sets of all the players in $\mathcal{N}$ except $i$, and $\mathcal{X} \defeq \prod_{i\in\mathcal{N}} \mathcal{X}_i$ as the joint feasibility set. We denote $x \defeq (x_i)_{i\in\mathcal{N}}$ as the collective strategy of all the players. Throughout the paper $\nabla$ will stand for the total derivative, $\nabla_i$ the partial derivative w.r.t. $x_i$ and $\nabla_{-i}$ refers to all the terms $(\nabla_{j})_{j\in\mathcal{N}\setminus\{i\}}$. For second-order derivatives we use the notation: $\nabla_{i,j}$ or $\nabla_i^2$, while we let $D$ be the gradient. Finally, $\|\cdot\|$ stands for the $L_2$ norm and $C^d$ refers to the class of continuously differentiable functions of order $d$.

\section{Conjectural Stackelberg Games}
\label{sec:conj_stack}
In CV games, each player builds conjectures about the other players' best responses. Two scenarios are considered in the literature. In the first scenario, each player $i \in \mathcal{N}$ assumes that the other players react homogeneously to its strategy thus giving rise to the modified objective function $f_i(x_i,\bar{x}),\,\forall i \in \mathcal{N}$ where $\bar{x} \defeq \sum_{j\in\mathcal{N}\setminus\{i\}} x_j$. In the second scenario, player $i$ relies on its decision variable $x_i$ to form conjectures
\begin{equation*}
    \gamma_i^j:\mathcal{X}_i\rightarrow\mathcal{X}_j,\,\forall j \in \mathcal{N} \setminus \{i\}.
\end{equation*}
The first scenario is the one discussed the most in the literature, since it can be simplified as a two-player game, and is easy to analyze. We instead focus on the second scenario which has more interest from a practical point of view. Before defining our game, we explicit our standing assumptions.
\begin{assumptions}
\label{ass:obj}
We consider the following assumptions for the game:
\begin{enumerate}
    \item $f_i(\cdot)\in C^{d}, \forall i\in\mathcal{N}$, with $d \geq 2$;
    \item $f_i(\cdot),\forall i \in\mathcal{N}$ is $M_1-$Lipschitz and convex;
    \item The set $\mathcal{X}_i$ is compact and convex.
\end{enumerate}
\end{assumptions}
We now present some standard definitions from the CV game literature.
\begin{definition}[Conjectural Variations (CV) Game \cite{baar_dynamic_1998}]
    \label{def:cvg}
    Given a set $\mathcal{N}$ of players, each with decision variables $x_i\in\mathcal{X}_i\subseteq\mathbb{R}^{m_i}$ and objective function $f_i(x_i,x_{-i})$, we define a conjecture as the function $\gamma_i^j:\mathcal{X}_i\rightarrow\mathcal{X}_j,\forall j\in \mathcal{N}\setminus\{i\}$. A conjecture variations (CV) game is defined as:
    \begin{equation}
        \forall i \in \mathcal{N}, \quad \min_{x_i}\{f_i(x_i,\Tilde{x}_{-i})\,|\,\Tilde{x}_{-i}=(\gamma_i^j(x_i))_{j\in\mathcal{N}\setminus\{i\}}\}.
    \end{equation}
\end{definition}
The outcome of CV games can be analyzed relying on CVE. We let $\{\gamma_i^j(x_i)\}_{j\in\mathcal{N}\setminus\{i\},\,i\in\mathcal{N}}$ be the stack of all the players' conjectures. 
\begin{definition}[Conjectural Variations Equilibrium \cite{figuieres_theory_2004}]
\label{def:cve}
    Given the conjectures $\{\gamma_i^j(x_i)\}_{j\in\mathcal{N}\setminus\{i\},\,i\in\mathcal{N}}$, the collective strategy $x^c\defeq(x_i^c)_i$ is a Conjectural Variations equilibrium (CVE) if $x^c$ is a solution of the optimization problem
    \begin{equation}
        \forall\;i\in\mathcal{N},\;\;\min_{x_i}\{f_i(x_i,\Tilde{x}_{-i}) \quad | \quad \Tilde{x}_{-i}=(\gamma_i^j(x_i))_{j\neq i}\}.
    \end{equation}
\end{definition}
Definition~\ref{def:cve} may generate a multiplicity of solutions, since it is always possible to find a set of conjectures for any given point such that it becomes a CVE \cite{figuieres_theory_2004}. Notice that this also entails that for a given point in the game, which would provide a better result for the players (closer to the social welfare), but it is not stable, we may obtain a set of conjectures that \textit{transforms} it into a stable point.

In order to deal with the multiplicity of solutions, the concept of consistency has been introduced as a refinement of a CVE.
\begin{definition}[Consistent Conjectural Variations Equilibrium \cite{bresnahan_duopoly_1981}]
\label{def:ccve}
    Given the conjectures $\{\gamma_i^j(x_i)\}_{j\in\mathcal{N}\setminus\{i\},\,i\in\mathcal{N}}$, the collective strategy $x^c\defeq (x_i^c)_i$ forms a Consistent Conjectural Variations Equilibrium (CCVE) if
    \begin{enumerate}
        \item $x^c$ is a CVE for the conjectures $\gamma_i^j(x_i), \forall j\in \mathcal{N}\setminus\{i\},\,\forall i\in\mathcal{N}$;
        \item There exists $\epsilon > 0$ such that $\nabla\gamma_i^j(x_i) = \nabla_i x_j^c(x_{-j})$,$\forall\;x_i\in\mathcal{X}_i$ such that $\|x_i^c-x_i\|<\epsilon, \, \forall i,j\in\mathcal{N}, i\neq j$;
    \end{enumerate}
    where $x_i(x_{-i})$ denotes player $i$'s best-response strategy.
\end{definition}
The existence of CCVE is not easy to establish \cite{figuieres_theory_2004}. For instance, duopolies with quadratic marginal costs may not allow for polynomial or symmetric analytic consistent conjectures. Similarly, \cite{calderone_consistent_2023} show that for quadratic games with affine conjectures, the existence of CCVE boils down to finding solutions to coupled asymmetric Riccati equations. Positive results for the existence require, in general, specific settings; such results exist, e.g., for specific cases of Cournot's duopoly and voluntary contributions to a public good \cite{figuieres_theory_2004}, and in electricity markets where uniqueness is also considered \cite{liu_existence_2007}.

\subsection{Multi-Leader Single-Follower Stackelberg Games}
\label{sec:mlsf}
We consider Stackelberg games involving a set $\mathcal{N}$ of leaders and a single follower. The $N$ leaders act simultaneously, anticipating the reaction of the follower; then the follower replies after having observed the decision variables of all the leaders.
We formally define a Stackelberg game as follows:
\begin{definition}[Multi-Leader Single-Follower Stackelberg Game]
\label{def:sg}
    Consider a set $\mathcal{N}$ of leaders and a follower, with respective decision variables $x_i\in\mathcal{X}_i\subseteq\mathbb{R}^{m_i},\forall i \in \mathcal{N}$, and $y\in\mathcal{Y}\subseteq\mathbb{R}^{m_y}$, and objective functions: $f_i(x_i,x_{-i},y),\forall i \in \mathcal{N}$, and $g(x,y)$. The Stackelberg game can be described as follows:
\begin{subequations}
    \begin{align}
        \forall i \in \mathcal{N}, &\min_{x_i\in\mathcal{X}_i,y\in\mathcal{Y}} \, f_i(x_i,x_{-i},y),\\
        & \quad \mbox{\text{s.t.}} \quad y\in\argmin_{y\in\mathcal{Y}}g(x,y).
    \end{align}
\end{subequations}
\end{definition}
\begin{assumptions}
\label{ass:follower_game}
We need additional assumptions on the follower's problem:
\begin{enumerate}
    \item The follower has a unique solution once $x\in\mathcal{X}$ is fixed;
    \item The follower's best response, $y(x)$, is $M_2-$Lipschitz;
    \item $g(\cdot)\in C^{d'}$ with $d'\geq 2$.
\end{enumerate}
\end{assumptions}
We can see how the hierarchy introduced in Definition~\ref{def:sg} already implicitly uses the concept of conjectures: the leaders need to act before the follower, meaning that they have to use an approximation, or a \textit{conjecture}, in place of the actual strategy of the follower. If we extend the anticipation to all players, meaning that the leaders have conjectures about the follower's strategy and also about the other leaders, we obtain a new class of games, that we call Conjectural Stackelberg Games.
Considering Definitions~\ref{def:cvg} and \ref{def:sg}, a Conjectural Stackelberg Game can be formalized as follows:
    \begin{subequations}
    \begin{align}
        \forall i \in \mathcal{N}, &\min_{x_i\in\mathcal{X}_i} \, f_i\left(x_i,(\gamma_i^j(x_i))_{j\neq i},\gamma_i^y(x_i)\right),\\
        & \min_{y\in\mathcal{Y}}g(x,y).
    \end{align}
    \label{eq:csg}
\end{subequations}
In Figure~\ref{fig:game_diag}, we represent schematically the structure of these games. The $N$ leaders receive feedback from their conjectures and update their strategies based on that information; then the follower optimizes its strategy after having observed the leaders' strategies.
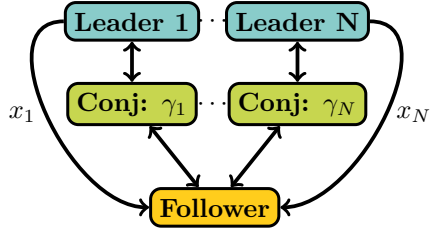
\begin{figure}[t]
        \centering
        \begin{tikzpicture}[
    box/.style={rectangle, draw, rounded corners, minimum width=1cm, minimum height=0.5cm, line width=0.4mm},
    leader/.style={box, fill=bleuclairINRIA},
    conj/.style={box, fill=vertclairINRIA},
    follower/.style={box, fill=jauneINRIA},
    scale=0.55
]

\node[leader] (L1) at (0,2) {\textbf{Leader 1}};
\node[leader] (LN) at (4,2) {\textbf{Leader N}};

\node[conj] (C1) at (0,0) {\textbf{Conj: }$\gamma_1$};
\node[conj] (CN) at (4,0) {\textbf{Conj: }$\gamma_N$};

\node[follower] (F) at (2,-2.5) {\textbf{Follower}};

\draw[<->, line width =0.5mm] (L1) -- (C1);
\draw[<->, line width =0.5mm] (LN) -- (CN);
\draw[<->, line width =0.5mm] (F) -- (C1);
\draw[<->, line width =0.5mm] (F) -- (CN);
\draw[<-, line width =0.5mm, out=180,in=180, looseness=1] (F.west) to node[left]{$x_1$} (L1.west);
\draw[<-, line width =0.5mm, out=0,in=0, looseness=1] (F.east) to node[right]{$x_N$} (LN.east) ;

\node at (2,2) {$\cdots$};
\node at (2,0) {$\cdots$};

\end{tikzpicture}
\caption{Schematic representation of a Conjectural Stackelberg Game.}
\label{fig:game_diag}
\end{figure}
Following the literature on conjectural games, we introduce the novel concept of Conjectural Stackelberg Equilibrium (CSE):
\begin{definition}[Conjectural Stackelberg Equilibrium]
\label{def:cse}
    Given the conjectures $\{\{\gamma_i^j(x_i)\}_{j\in\mathcal{N}\setminus\{i\}},\gamma_i^y(x_i)\}_{i\in\mathcal{N}}$, $(x^*,y^*)$ is a Conjectural Stackelberg Equilibrium (CSE) if it is a solution of Eqs.~\eqref{eq:csg}. 
\end{definition}
Notice that Definition~\ref{def:cse} does not consider the notion of consistency. Following the same steps as for standard conjectural games, we define requirements for CSE to be consistent:
\begin{proposition}[Consistency for CSE]
\label{prop:ccse}
    Given the conjectures $\{\{\gamma_i^j(x_i)\}_{j\in\mathcal{N}\setminus\{i\}},\gamma_i^y(x_i)\}_{i\in\mathcal{N}}$, a CSE $(x^*,y^*)$ is a Consistent CSE (CCSE) if
    \begin{enumerate}
        \item There exists $\epsilon_L>0$ such that $\nabla\gamma_i^j(x_i) = \nabla_i x_j^c(x_{-j})$,$\forall\;x_i\in\mathcal{X}_i$; $||x_i^*-x_i||<\epsilon_L, \forall i,j\in\mathcal{N}, i\neq j$;
        \item There exists $\epsilon_F>0$ such that $\nabla\gamma_i^y(x_i)=\nabla_i y(x)$, $\forall\;x_i\in\mathcal{X}_i$, $||x_i^*-x_i||<\epsilon_F, \forall i\in\mathcal{N}$ and $\forall y\in\mathcal{Y}$, $||y^*-y||<\epsilon_F$.
    \end{enumerate}
\end{proposition}
The set of CCSEs is a superset of the set of Stackelberg Equilibria (SE), since if the conjectures coincide with the best responses, the game outcome coincides with a SE. 

\paragraph{Differential Equilibrium Definition}
Given conjectures $\{\{\gamma_i^j(x_i)\}_{j\in\mathcal{N}\setminus\{i\}},\gamma_i^y(x_i)\}_{i\in\mathcal{N}}$, we define first and second-order conditions for an equilibrium $(x^*,y^*)$ to be a CSE
\begin{equation*}
    \begin{aligned}
&\nabla f_i\left(x_i^*,(\gamma_i^j(x_i^*))_{j\neq i},\gamma_i^y(x_i^*)\right) = 0, \\
&\nabla^2 f_i\left(x_i^*,(\gamma_i^j(x_i^*))_{j\neq i},\gamma_i^y(x_i^*)\right) > 0,\, \forall i\in\mathcal{N},
\end{aligned}
\end{equation*}
and $\nabla_y g(x^*,y^*) = 0, \, \nabla_y^2\,g(x^*,y^*) > 0$. Notice that under Assumptions~\ref{ass:obj} and \ref{ass:follower_game}, this definition is equivalent to Definition~\ref{def:cse}.

\subsection{Form of the Conjectures}
Our approach is \textit{conjecture-agnostic} since we do not impose any priors on the form for the conjectures. 
\begin{assumptions}
\label{ass:conj}
    The conjecture functions $\gamma_i^j(\cdot),\; \forall i,j \in \mathcal{N}, i\neq j$ are twice differentiable.
\end{assumptions}
Some possible forms for these functions, that we will consider throughout the paper are the following:
\begin{itemize}
    \item \textit{Affine conjectures}: $\gamma_i^j(x_i)=a_ix_i+b_i$, for $a_i, b_i \in \mathbb{R}$;
    \item \textit{Polynomial conjectures}: $\gamma_i^j(x_i)=\sum_{g=0}^Gc_{i,g}x_i^g$, for $c_{i,g}\in\mathbb{R}$;
    \item \textit{Neural networks}: $\gamma_i^j(x_i)=\Phi_i^j(x_i)$, where $\Phi_i^j$ is the architecture for the network used.
\end{itemize}

\section{Game Analysis}
\label{sec:analysis}
We can now proceed with a theoretical analysis of the game, comparing conjectural variations and Stackelberg games.  To that purpose, we equip $\mathcal{X}$ and $\mathcal{Y}$ with distances $d_{\mathcal{X}}$ and $d_{\mathcal{Y}}$ respectively. For $x^1,\,x^2\in\mathcal{X}\times\mathcal{Y}$, where $x^i=(x^i,y(x^i)), i=1,2$ we define $d(x^1,x^2)\defeq || \left(d_{\mathcal{X}}(x^1,x^2),\allowbreak d_{\mathcal{Y}}(y(x^1),y(x^2))\right)||$ as the distance of the projections of $x^1, x^2$ on $\mathcal{X}$ and $\mathcal{Y}$ respectively.

\subsection{Bounding the Objective Functions}
\begin{proposition}[Objective Function Bound]
\label{prop:bound}
    Suppose Assumptions \ref{ass:obj}-\ref{ass:conj} hold. Consider a SE $x^s\defeq (x^*,y^*(x^*))$ and a CSE $x^c \defeq (x^{\gamma},y^*(x^{\gamma}))$, then the distance between the objective function of any leader $i \in \mathcal{N}$ at a Stackelberg equilibrium and at a conjectural variations Stackelberg equilibrium is upper bounded
    \begin{equation}
    \|f_i(x^*,y^*(x^*))-f_i(x^{\gamma},y^*(x^{\gamma}))\|\leq R d_{\mathcal{X}}(x^*,x^{\gamma}),
\end{equation}
with $R \defeq M_1\sqrt{1+M_2^2}$.
\end{proposition}
\begin{proof}
    Let $x^c$ correspond to the decision variable obtained at equilibrium by the leaders considering their conjectures, but evaluated with the \textit{actual} follower's response $y^*(x^{\gamma})$. Using the assumption of Lipschitz continuity of $f_i$ from Assumption \ref{ass:obj}, we obtain:
\begin{equation}
\label{eq:bound_1}
    \|f_i(x^s)-f_i(x^c)\|\leq M_1 d(x^s,x^c),
\end{equation}
which is a first bound for the distance between the objective functions evaluated in $x^s$ and $x^c$.
If we explicit the right-hand side of Eq.~\eqref{eq:bound_1}, we obtain $$d(x^s,x^c)=\sqrt{d_{\mathcal{X}}(x^*,x^{\gamma})^2+d_{\mathcal{Y}}(y^*(x^*),y^*(x^{\gamma}))^2}.$$ 
The second term, considering the distance between the follower's best-responses, can be developed further using the Lipschitz continuity of the follower's best-response from Assumption \ref{ass:follower_game}
\begin{equation}
\label{eq:bound_11}
   d_{\mathcal{Y}}(y^*(x^*),y^*(x^{\gamma}))\leq M_2 d_{\mathcal{X}}(x^*,x^{\gamma}). 
\end{equation}
Combining Eq.~\eqref{eq:bound_1} and Eq.~\eqref{eq:bound_11} concludes the proof.
\end{proof}

\subsection{Reaching a Stackelberg Equilibrium}

The previous result linked a Stackelberg equilibrium with a conjectural variations Stackelberg equilibrium. Now, we aim to define connections between Stackelberg equilibrium and consistent conjectural variations Stackelberg equilibrium. Using the differential equilibrium definition from Section \ref{sec:conj_stack}, we write explicitly the full gradient of leader $i$ considering either a SE or a CCSE. First, in $x^s = (x^*,y^*)$, we get
\begin{equation}
\begin{split}
    \nabla f_i(x^s) &= \nabla_i f_i(x^s)+\nabla_{-i}f_i(x^s)\nabla_i x_{-i}^*+\nabla_yf_i(x^s)\nabla_iy^*.
\end{split}
    \label{eq:se_grad}
\end{equation}
Second, in a CCSE, in $x^c = (x_i^{\gamma},\gamma_i^{-i}(x_i^{\gamma}),\gamma_i^y(x_i^{\gamma}))$, and using $\gamma_i^{-i}(x_i)=(\gamma_i^j(x_i))_{j\neq i}$, we have
\begin{equation}
\begin{split}
    \nabla f_i(x^c) &=\nabla_i f_i(x^c)+\nabla_{-i}f_i(x^c)\nabla_i \gamma_i^{-i}(x_i^{\gamma})+\nabla_yf_i(x^c)\nabla_i\gamma_i^y(x_i^{\gamma}).
\end{split}
    \label{eq:ccse_grad}
\end{equation}
We note that Eq.~\eqref{eq:se_grad} and Eq.~\eqref{eq:ccse_grad} are equal if, and only if
\begin{subequations}
    \begin{align*}
        & \nabla_i f_i(x^s)- \nabla_i f_i(x^c) = 0,\\
        &\nabla_{-i}f_i(x^s)\nabla_i x_{-i}^* - \nabla_{-i}f_i(x^c)\nabla_i \gamma_i^{-i}(x_i^{\gamma}) = 0,\\
        &\nabla_yf_i(x^s)\nabla_i y^* - \nabla_yf_i(x^c)\nabla_i\gamma_i^y(x_i^{\gamma}) = 0.
    \end{align*}
\end{subequations}
Since in a CCSE the variations of the conjectures must be equal to the real best-response variations, we can focus on the terms related to the derivatives of the objective function.
Considering the derivative w.r.t. variable $i$ we can rewrite the term $\nabla_i f_i(x^s)- \nabla_i f_i(x^c)$ as $\nabla_i(f_i(x^s)-f_i(x^c))$, and given Assumptions~\ref{ass:obj} for $f_i$, using the mean value theorem we obtain
\begin{equation}
    \nabla_i[f_i(x^s)-f_i(x^c)] = \nabla_i[D f_i((1-\tau)x^c+\tau x^s)(x^s-x^c)^{\top}]
    \label{eq:mvt}
\end{equation}
with $\tau\in[0,1]$. We set $h \defeq (1-\tau)x^c+\tau x^s$, 
then expanding the right-hand side of Eq.~\eqref{eq:mvt} we obtain the following term:
\begin{equation}
\begin{split}
    \nabla_i[f_i(x^s)-f_i(x^c)] =& \nabla_i^2f_i(h)(x_i^{\gamma}-x_i^*) +\nabla_{-i,i}f_i(h)(\gamma_i^{-i}-x_{-i}^*)\\
    +&\nabla_{y,i}f_i(h)(\gamma_i^y-y^*),
\end{split}
\label{eq:i_term}
\end{equation}
and, with the same procedure, we equivalently obtain similar terms for the derivatives w.r.t. $-i$ and $y$:
\begin{equation}
\begin{split}
    \nabla_{-i}[f_i(x^s)-f_i(x^c)] = &\nabla_{i,-i}f_i(h)(x_i^{\gamma}-x_i^*)+\nabla_{-i}^2f_i(h)(\gamma_i^{-i}-x_{-i}^*)\\
    +&\nabla_{y,-i}f_i(h)(\gamma_i^y-y^*),
\end{split}
\label{eq:-i_term}
\end{equation}
\begin{equation}
\begin{split}
    \nabla_y[f_i(x^s)-f_i(x^c)] = &\nabla_{i,y}f_i(h)(x_i^{\gamma}-x_i^*)+\nabla_{-i,y}f_i(h)(\gamma_i^{-i}-x_{-i}^*)\\
    +&\nabla_y^2f_i(h)(\gamma_i^y-y^*).
\end{split}
\label{eq:y_term}
\end{equation}
Comparing Eqs.~\eqref{eq:i_term}~-~\eqref{eq:y_term} we obtain the following link between SE and CCSE:
\begin{proposition}[Existence of Conjectures]
    Given a SE $x^s$, if each player's best-response is differentiable, then we can always find a set of conjectures to reach $x^s$.
\end{proposition}

\begin{proof}
It suffices to choose the best-responses as the conjectures.
\end{proof}

We can further analyze Eqs.~\eqref{eq:i_term}~-~\eqref{eq:y_term} from the point of view of the objective functions, obtaining two interesting cases:
\begin{enumerate}
    \item \textbf{Linear Objective Functions:} Given linear objectives functions for the leaders $f_i(x,y) = a_ix_i+\sum_{j\neq i}b_{i,j}x_j+c_iy$, we observe that all the second-order derivatives in Eqs.~\eqref{eq:i_term}~-~\eqref{eq:y_term} vanish, meaning that SE and CCSE coincide. This means that the full derivative of the conjectural game for leader $i$ becomes:
    \begin{equation*}
        \nabla_i f_i(x,y) = a_i + \sum_{j\neq i}b_{i,j}\nabla_i\gamma_i^j(x_i^{\gamma}) + c_i\nabla_i\gamma_i^y(x_i^{\gamma}),
    \end{equation*}
    which coincides with the SE definition, given the conditions of Proposition~\ref{prop:ccse} on the conjectures.
    \item \textbf{Quadratic Objective Functions:} Consider quadratic objective functions such as: $f_i(x,y)=\boldsymbol{z}A^i\boldsymbol{z}^{\top}$, where $\boldsymbol{z}=(x,y)$ and $A^i$ is the coefficients matrix for leader $i$, with dimensions $(m+m_y)\times(m+m_y)$. In Eqs.~\eqref{eq:i_term}~-~\eqref{eq:y_term} all the second derivatives become constant terms, i.e., the term $\nabla_i^2f_i(m)$ becomes $A_{i,i}^i$, the term $\nabla_{y,i}f_i(m)$ becomes $A_{i,y}^i$, where the subscripts $i,j$ are the row and column indexes of the matrix $A$. This implies that the difference between the equilibria depends solely on how close the conjectures at a CCSE are to the other players' best responses at a SE.
\end{enumerate}
Note that the conjecture parameters are generally not known and need to be learned. To that purpose, we propose a two-stage algorithm in the next section.

\section{The \texttt{COSTAL} Algorithm}
\label{sec:algo}
In our two-stage algorithm, we first train the leaders' conjecture models, and then let the players play the Stackelberg game relying on the learned conjectures.

\subsection{Training Conjectures}
\begin{algorithm}[tb]
\caption{Pseudo-code for training the conjectures.}
\label{alg:conj}
\begin{algorithmic}[1] 
\State \textit{Input}: $\sigma, T, |\mathcal{B}|$
\For{$t\in[0,T]$}
\State Sample $x^t\in\mathcal{X}, y^t\in\mathcal{Y}, \xi\sim\mathcal{N}(0,\sigma)$
\State Compute $\tilde{y}(x^t) = y^*(x^t)+\xi$
\For{$i\in\mathcal{N}$}
\State Sample $\xi\sim\mathcal{N}(0,\sigma)$
\State Compute $\tilde{x}_i(x_{-i}^t,y^t) = x_i^*(x_{-i}^t,y^t)+\xi$
\EndFor
\For{$i \in\mathcal{N}$}
\State Save $(x_i^t,\tilde{y})$ in $\mathcal{D}_i^y$
\For{$j\in\mathcal{N}, j\neq i$}
\State Save $(x_i^t,\tilde{x}_j)$ in $\mathcal{D}_i^j$
\EndFor
\EndFor
\EndFor
\State Train the conjectures on the data sets $\mathcal{D}_i^j,\mathcal{D}_i^y,\forall i,j\in\mathcal{N}$
\State \textit{Output}: trained conjectures $\gamma_i^j,\gamma_i^y,\forall i,j\in\mathcal{N}$
\end{algorithmic}
\end{algorithm}
Before the game is played, a training phase takes place during which the leaders collect data about the other players. This can be seen as a similar process to \textit{cheap talk} \cite{farrell_talk_1995}, or Generative Adversarial Networks \cite{farnia_gans_2020}, where a pre-trained discriminator is used.
At each step of the training phase we sample uniformly a different decision variable for each leaders, then the best response for each player is computed, considering the randomly sampled variables and adding noise: $\forall i \in \mathcal{N},\,\tilde{x}_i = x_i^*(x_{-i},y) + \xi $ with $\xi\sim\mathcal{N}(0,\sigma)$. The tuples $(x_i,\tilde{x}_j),\forall j\in \mathcal{N}\setminus\{i\}$, and $(x_i,\tilde{y}),\,\forall i \in \mathcal{N}$, where $\tilde{y}$ is the noisy follower's reaction to the strategy of the leaders, are stored in datasets $(\mathcal{D}_i^j)_{i,j}$ and $(\mathcal{D}_i^y)_i$. This process is repeated $T$ times, creating for each player a dataset of size $T$ for every other player in the game.
The next step is then to update the conjectures to learn the other players' behaviors over these data streams using stochastic gradient descent with batch size $|\mathcal{B}|$, and with a quadratic loss defined as follows:
\begin{equation*}
    L_i^j(\gamma_i^j) \defeq \frac{1}{|\mathcal{B}|}\sum_{b\in\mathcal{B}\subset \mathcal{D}_i^j}\left(x_j^b-\gamma_i^j(x_i^b)\right)^2, \forall i,j \in \mathcal{N},\, j\neq i.
\end{equation*}
The algorithm to learn the players' conjectures is described in Algorithm~\ref{alg:conj}. 
In the sampling part (lines 3-7), the operations scale as $T\cdot N$, thus it is linear in the number of agents. In the training part, we have to train $N\cdot(N-1)$ conjecture functions, thus it scales quadratically with the number of agents. Notice that this part could be run in parallel for each agent (or by groups of agents), if $N$ becomes large.

\subsection{Learning Conjectural Stackelberg Equilibria}
\label{sec:proof}
\begin{algorithm}[tb]
\caption{Pseudo-code for learning a Conjectural Stackelberg Equilibrium.}
\label{alg:strat}
\begin{algorithmic}[1]
\State \textit{Input}: $T, \eta, \{\gamma_i^j,\gamma_i^y,\forall i,j\in\mathcal{N}\}, \{f_i(\cdot), \forall i\in\mathcal{N}\}$
\State Initialize $x^0\in\mathcal{X}$
\For{$t\in[0,T]$}
\For{$i\in\mathcal{N}$}
\State Compute $D_i = \nabla_if_i$
\State Compute $D_{-i} = \nabla_{-i}f_i\nabla_i\gamma_i^{-i}$
\State Compute $D_y = \nabla_yf_i\nabla_i\gamma_i^y$
\State $x^{t+1}_i = x_i^t - \eta(D_i+D_{-i}+D_y)$
\EndFor
\EndFor
\State \textit{Output}: Final strategies $\{x_i^T\;\forall i\in\mathcal{N}\},y^*(x^T)$
\end{algorithmic}
\end{algorithm}

The leaders participating in the Stackelberg game update their strategies through gradient descent. Considering the continuous counterpart of our system, the players update their strategies as follows:
\begin{subequations}
    \begin{align}
        &\dot{x}_i = -\nabla f_i(x_i,x_{-i},y),\;\forall\;i\in\mathcal{N},\\
        &\dot{y} = -\nabla_y g(x,y).
    \end{align}
    \label{eq:std_system}
\end{subequations}
The main issue with this formulation is that, since all the leaders are coupled, computing the full gradient for each of them is not feasible. 
For the settings we consider, the presence of conjectures helps at rewriting this system in a simpler way. Substituting the conjectures in  the objective function of leader $i$, $f_i$, we obtain $f_i(x_i,(\gamma_i^j(x_i))_{j\neq i},\gamma_i^y(x_i))$, which depends only on $x_i$. We define $\Tilde{f}_i(x_i) \defeq f_i(x_i,(\gamma_i^j(x_i))_{j\neq i},\gamma_i^y(x_i))$, with $\Tilde{f}_i$ the conjectured objective function. Eqs.~\eqref{eq:std_system} can then be rewritten as follows:
\begin{subequations}
    \begin{align}
        &\dot{x}_i = -\nabla_i \tilde{f}_i(x_i),\;\forall\;i\in\mathcal{N}, \label{eq:leader_cont}\\
        &\dot{y} = -\nabla_y g(x,y).
    \end{align}
    \label{eq:conj_system}
\end{subequations}
We now explicit $\nabla_i\tilde{f}_i$:
\begin{equation}
\begin{split}
        \nabla_i\tilde{f}_i(x_i) = &\nabla_if_i(x_i,\boldsymbol{\gamma}_i^{-i},\gamma_i^y) + \nabla_{-i}f_i(x_i,\boldsymbol{\gamma}_i^{-i},\gamma_i^y)\nabla_i\boldsymbol{\gamma}_i^{-i}\\
        +&\nabla_yf_i(x_i,\boldsymbol{\gamma}_{-i},\gamma_i^y)\nabla_i\gamma_i^y.
\end{split}
\label{eq:cont_grad}
\end{equation}
The key feature of Eq.~\eqref{eq:cont_grad} is that it is easily obtained once the conjectures are given. Furthermore, a convergence guarantee for the system described in Eqs.~\eqref{eq:conj_system} can be obtained following  \cite{fiez_implicit_2020} and \cite{borkar_stochastic_2008}. Specifically, we can write the discrete version of this update rule for leader $i$:
\begin{equation}
    x_i^{t+1} = x_i^{t} - \eta^{t}(\nabla_i\tilde{f}_i(x_i^t)+\zeta_i^t),
    \label{eq:disc_update}
\end{equation}
where we also consider the noise $\zeta$ deriving from the fact that the leader has access to an unbiased estimate of the gradient. This allows us to use well-known results from stochastic approximation theory \cite{borkar_stochastic_2008}, from which we also state the following assumptions:
\begin{assumptions}[{ \cite[Sec. 2.1, Assumptions A1-A2-A3]{borkar_stochastic_2008} }]
\label{ass:lipschitz}
For all leader $i \in \mathcal{N}$, $\nabla_i\tilde{f}_i$ is Lipschitz continuous and $\|\nabla_i\tilde{f}_i\|<+\infty$, $\{\zeta_i^t\}$ is a martingale difference sequence, and the gradient step sizes are chosen such that $\sum_t \eta^t=\infty, \, \sum_t(\eta^t)^2<+\infty$.
\end{assumptions}

\begin{theorem}
    Suppose Assumptions~\ref{ass:lipschitz} hold, then $\{(x_i^t)_i\}_t$ converges almost surely to a local CSE $(x_i^{\gamma})_i$ of Eqs.~\eqref{eq:conj_system}.
\end{theorem}
\begin{proof}
    We consider the differential CSE $(x^{\gamma},y^{\gamma})$, given the set of conjectures $\{(\gamma_i^j(x_i))_{j\neq i},\gamma_i^y(x_i);i\in\mathcal{N}\}$.
    From the differential version of Definition~\ref{def:cse}, we know that for any leader $i$, $\nabla_i \tilde{f}_i(x_i^{\gamma}) = 0$ and $\nabla_i^2\tilde{f}_i(x_i^{\gamma}) > 0$, meaning that it is a stable point for the differential equation: $\dot{x}_i = -\nabla_i \tilde{f}_i(x_i)$. 
    The discrete update $x_i^{t+1} = x_i^{t} - \eta^{t}(\nabla_i\tilde{f}_i(x_i^t)+\zeta_i^t)$ is then a stochastic approximation of the continuous process, tracking the ODE asymptotically. Given the assumptions we described in the statement of the theorem and using \cite[Theorem 2, Corollary 4]{borkar_stochastic_2008}, the discrete sequence of each leader $i$ converges almost surely to a compact internally chain transitive set of the ODE. This is a a closed compact set $\mathcal{A}\subset \mathcal{X}_i$ such that: for any trajectory $x_i(t)$ with $x_i(0)\in\mathcal{A}$, $x(t)\in\mathcal{A}\,\forall t\geq 0$; and for any $x^a,x^b\in\mathcal{A}, \epsilon > 0, T>0$, there exists $n\geq1$ and $x^a=x^0,x^1,\cdots,x^n=x^b$ such that the trajectory of the ODE starting at $x^l, 0\geq l>n$, reaches the $\epsilon$-neighbourhood of $x^{l+1}$ after a time $\geq T$. Finally, since any stable attractor of the dynamic has to satisfy $\nabla_i^2\tilde{f}_i > 0$ for all leaders \cite{strogatz_nonlinear_2015}, and the followers always replies with the best response, by definition the only internally chain transitive sets are CSEs. We can then conclude that the algorithm using Eq.~\eqref{eq:disc_update} converges a.s. to a CSE of Eqs.~\eqref{eq:conj_system}.
\end{proof}

In Algorithm~\ref{alg:strat} we report the structure of the algorithm in pseudo-code. Notice that the update of the strategy of player $i$ (lines 5-8), obtained from Eq.~\eqref{eq:cont_grad}, only depends on variables of player $i$. This means that the loop over all the players $\mathcal{N}$ could be run in parallel instead than sequentially, speeding up \texttt{COSTAL} for large number of agents. Furthermore, this update does not even need to be run on the same machine, allowing full decentralization for the leaders. 
Note that \texttt{COSTAL} may never reach a global solution of Eqs.~\eqref{eq:std_system}. There exist integrated approaches based on the reformulation of the Stackelberg games as an equivalent mathematical program with constraints \cite{dempe_bilevel_2020} to compute global solutions, but they do not allow to account for the learning capabilities of the players. Further, reaching a global equilibrium can be expensive form a computational point of view, thus focusing on a local equilibrium might be preferable. In addition, from a system point of view, the adjustments of the players' strategies are usually small, meaning that it makes sense to study local solutions \cite{zheng_stackelberg_2021}. 
From \cite{bubeck_convex_2015}, we know that the convergence rate of the gradient descent given a smooth function, which we already require in Assumptions~\ref{ass:lipschitz}, is $O(1/T)$.
Note that the loop over agents in Algorithm~\ref{alg:strat} (lines 4-9) includes an inner loop (line 6), which runs over the $N-1$ remaining agents. This entails that the operations for each gradient step scale as $O(N^2)$. On the other hand, we mentioned that the gradient update can be run in parallel for all agents, bringing back the scaling to a linear behavior.

\section{Experiments}
\label{sec:experiments}
We present numerical results for the \texttt{COSTAL} algorithm applied to two stylized games; for both games, we report additional data and a full script to reproduce the experiments in the Appendix. We focus on comparing the learned strategies with known equilibria, to check whether a CSE could benefit the players of the game, and if so, with which conjectures. In particular we compare our algorithm against a naive gradient descent update, with full access to the other players' variables, described by:
\begin{equation}
    x_i^{t+1} = x_i^t - \eta\nabla_if_i(x_i^{t},x_{-i}^t,y^t),\,\forall t.
    \label{eq:GD}
\end{equation}
We will refer to this benchmark algorithm in the next sections as \texttt{GD}. Note that this method is not decentralized and cannot be run in parallel, as every update needs all the players' variables.

\subsection{Leader's Dilemma}
\begin{figure}[t]
\begin{subfigure}[t]{.32\linewidth}
    \includegraphics[width=\linewidth]{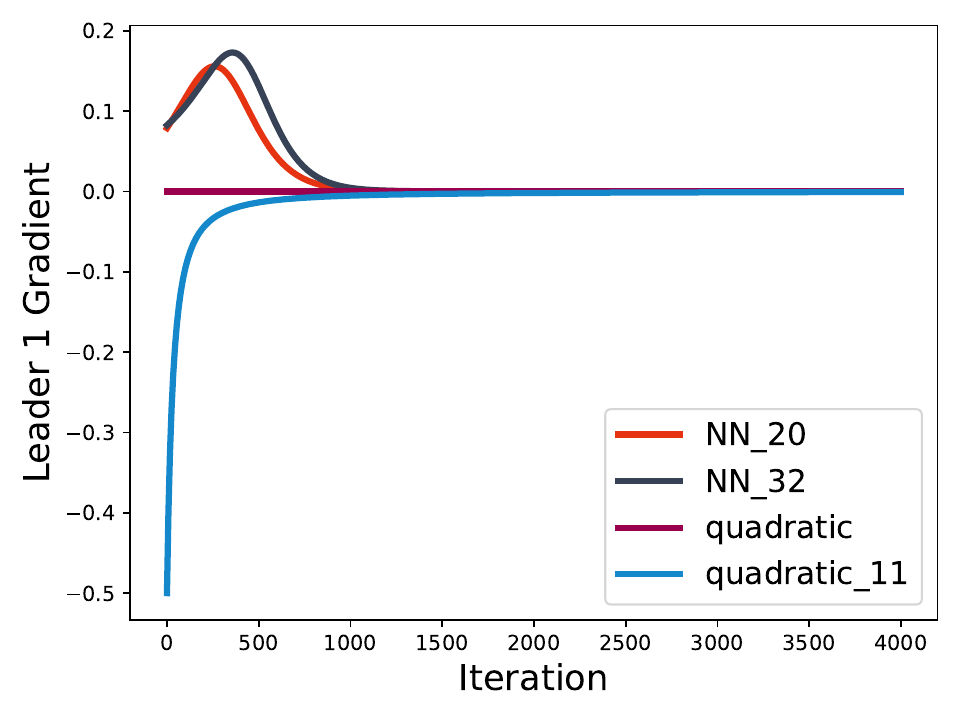}
\caption{Evolution of the gradient for leader 1.}
\label{fig:grad_x1_zonal}
\end{subfigure}
\hfill
\begin{subfigure}[t]{.32\linewidth}
    \includegraphics[width=\linewidth]{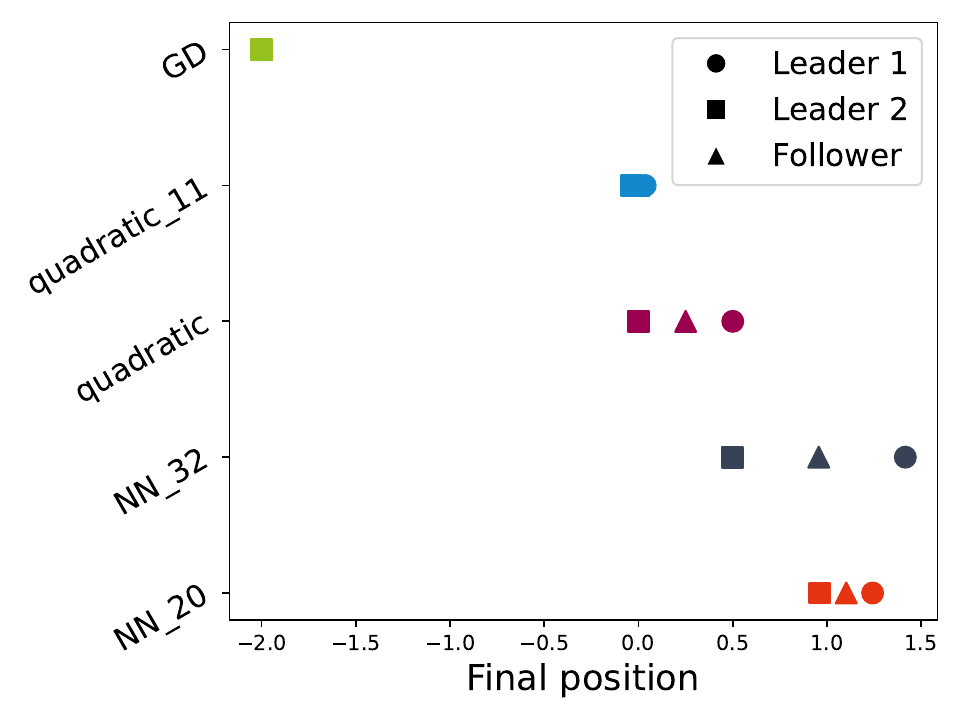}
\caption{Final position for all players.}
\label{fig:final_pos_zonal}
\end{subfigure}
\hfill
\begin{subfigure}[t]{.32\linewidth}
    \includegraphics[width=\linewidth]{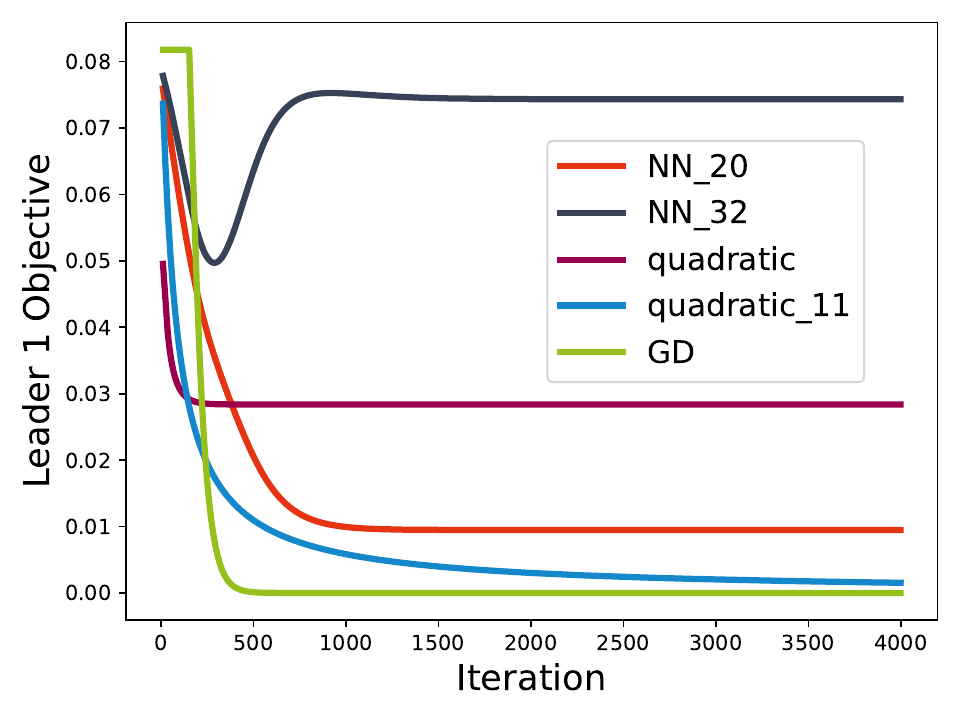}
    \caption{Evolution of the objective of leader 1.}
\label{fig:obj_x1_zonal}
\end{subfigure}
\caption{Results of the simulation for the game "Leader's Dilemma" with $K=-1.5$, we report only the results for one leader as the game is symmetric.}
\label{fig:zonal}
\end{figure}
We consider here a Stackelberg game where two leaders should try to cooperate by getting close to a target set by the follower. The objective function for the leaders are:
\begin{equation}
    f_i(x_i,x_j,y) = -(x_i-y)^2-K\left(1-e^{-(x_j-y)^2}\right),\,\forall i\in\{1,2\},
\end{equation}
and for the follower:
\begin{equation}
    g(x,y) = \left(\frac{x_1+x_2}{2} - y\right)^2.
\end{equation}
The leaders maximize their objective function, while the follower minimizes it, all variables are in the action space $[-2,2]$ and $K<-1$. It is easy to prove that the strategy $x_1=x_2$ for the leaders is a saddle point, thus they could obtain better payoffs by staying far from the target. In Fig.~\ref{fig:grad_x1_zonal}\footnote{The labels in the plots refer to which type of conjecture is used: NN stands for neural network, with the number referring to the size of the hidden layer; for the quadratic conjecture, just `quadratic' refers to $\gamma = x^2$, while `quadratic$\_$11' is $\gamma = x^2 + x$.} we check convergence speed of the algorithm: the leaders' gradient converges to 0 rather quickly for each class of conjectures, as we expected from our theoretical analysis. Furthermore, in Fig.~\ref{fig:obj_x1_zonal} we plot the leaders' objective function at equilibrium: we can see that the \texttt{GD} stays trapped in the saddle point, the two leaders stay in the same position, obtaining 0 as reward; the \texttt{COSTAL} algorithm with the different conjectures instead, finds stable solutions that guarantee a reward larger than 0, beating the \texttt{GD}. Finally, in Fig.~\ref{fig:final_pos_zonal}, we plot the final solutions reached in each case. We infer that the conjectures allow the leaders to find an equilibrium which is at a distance from the target set by the follower, while the \texttt{GD} method cannot.
We present more numerical results to give an intuition on why this happens in the Appendix.

\subsection{Revisiting Olsder's Paradox}
For our second experiment we take a game described in \cite{olsder_phenomena_2009} (players are maximizing). The game considers two players with action space $[0,+\infty]$ and with objective functions:
\begin{equation*}
    f_1(x_1,x_2) = (x_1-84)(-12.5x_1+21x_2+756),
\end{equation*}
\begin{equation*}
    f_2(x_1,x_2)= (x_2-50)(25x_1-50x_2+560).
\end{equation*}
This game can be interpreted as a paradox because it leads to a CCE with affine conjectures, which is more efficient than the NE and the SE (with player 1 as leader, and player 2 as follower) solutions of the game. This result is counter-intuitive, because it means that being bounded rational can be more profitable for the players and more efficient for a system point of view than being fully rational. The numerical results are reported in Tab.~\ref{tab:equilibria}, where we also highlight the social welfare optimum for the game, which is reached by maximizing the sum of the objective functions. For the numerical simulations, we consider two settings: simultaneous play (\texttt{N} at the front of the name in the graphs) and Stackelberg play (\texttt{S}), with player 1 using the algorithm as the leader, while player 2 replies with the best response. In Fig.~\ref{fig:olsder} we report the values reached by the objective function of player 1 and 2 respectively. We can see how using the \texttt{COSTAL} algorithm produces results that always beat the NE, and for player 1 in the simultaneous case we even achieve a better result than the SE. The CCE is far more profitable for the players and more efficient than the NE and the SE. However, this result needs to be mitigated in practice, because a large amount of information is needed to compute a CCE: relying on Proposition~\ref{prop:ccse}, player 1 would need to have access to player 2's objective function, together with its first and second-order derivatives. However, it is worth noticing that \texttt{COSTAL} enables to reach CSE without any information exchange among the players, achieving better performance (players' objective values, efficiency) than the NE and the SE. Furthermore, under the hierarchical play setting, we can observe that Proposition~\ref{prop:bound} holds: both players' objective functions are upper bounded by the SE.
\begin{table}[h]
    \centering
    \begin{tabular}{c c c c c}
        \toprule
         Solution &  $x_1$ & $x_2$ & $f_1$ & $f_2$\\
         \midrule
         \textbf{CCE} & $164.4$ & $81$ & $32320.8$ & $19220$ \\
         \textbf{SE} & $138.04$ & $65.11$ & $21411.6$ & $11415.8$\\
         \textbf{NE} & $123.98$ & $61.6$ & $19979.8$ & $6722.13$\\
         \textbf{SWO} & $300.04$ & $150.98$ & $38141.2$ & $56548.7$\\
         \bottomrule
    \end{tabular}
    \caption{Values at CCE, SE, NE and SWO for the Olsder's game.}
    \label{tab:equilibria}
\end{table}

\begin{figure}
    \centering
    \includegraphics[width=0.6\linewidth]{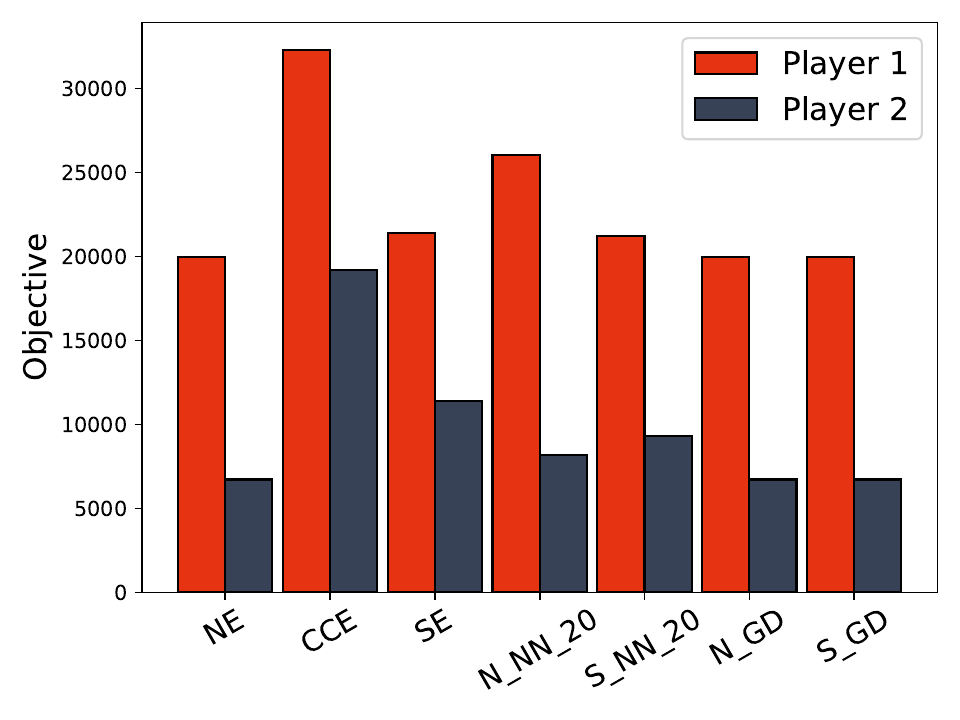}
    \caption{Final value of the objective function of both players of the Olsder's Paradox. }
    \label{fig:olsder}
\end{figure}

\section{Conclusions}
\label{sec:conc}
In this work we presented a new class of games, called Conjectural Stackelberg Games, aimed at bridging the gap between the use of conjectures and Stackelberg games. In particular we formalized how leaders may model opposite players' behavior.
To characterize the game we also propose an equilibrium definition, using the notion of consistency, first introduced in Conjectural Variations Games. As a first effort to study these new games, we compare the Conjectural Stackelberg Equilibrium with the standard Stackelberg Equilibrium for multiple settings.
We also want to stress the fact that we kept our whole formalism \textit{conjecture-agnostic}, meaning that we do not assume a specific form for the conjecture functions, which is in stark contrast with the existing literature.
Finally, we develop a multi agent learning algorithm (\texttt{COSTAL}) that can be used by the players of Conjectural Games to update their conjectures and their strategy. We also provide a proof of convergence for said algorithm.
We believe Conjectural Stackelberg Games could be a useful and interesting class of games to study further, both on the theoretical side, obtaining more results regarding their relation with Nash and Stackelberg equilibria; but also on the algorithmic side, using their structure to develop game-theory informed learning algorithm.

\bibliography{main}
\bibliographystyle{rlc}

\appendix

\section*{Appendix}
\label{sec:appendix}
We report here further data regarding the simulation discussed in the main paper, focusing on the trajectories of the players during the learning of their strategy and the convergence speed, which we can show empirically looking at the gradient at each step of the algorithm.

A full script to run all the experiments can be found at this \href{https://github.com/FrancescoMorri/COSTAL}{link}.

\paragraph{Leader's Dilemma}
For this game we already reported the gradient evolution and objective evolution for leader 1 in the main paper, so we show here in Fig.~\ref{fig:zonal_x2_grad} and Fig.~\ref{fig:zonal_x2_obj} the same plots but for leader 2, and is clearly visible the gradient quickly goes to 0 and the the results obtained by the \texttt{COSTAL} algorithm are superior to that of the basic \texttt{GD} approach. Furthermore, in Fig.~\ref{fig:zonal_strat} we report the evolution of the strategies of the two leaders, where the different results between our algorithm and the gradient descent method are even more distinct. Finally, in Fig.~\ref{fig:zonal_x1_x2} we report a comparison between the co-evolution of the two leaders and the Stackelberg equilibria of the game. In particular this game has a continuum of equilibria, between which the line $x_1=x_2$ is not stable and it also achieves a lower reward than $x_2=x_1\pm 2\sqrt{\log|K|}$. We can see that the \texttt{GD} method converge to the unstable point, while the equilibria reached by the conjectures are not on the Stackelberg Equilibria line, even though some are getting close. Once more this highlights the fact that the \texttt{COSTAL} algorithm converges to a stable equilibrium that is different from the standard notions of Nash/Stackelberg. A further explanation behind this can be obtained by looking at Fig.~\ref{fig:zonal_conj_obj}, where we plot the objective functions of the two leaders evaluated with the learned conjectures \texttt{NN$\_$32}: we can see that each player has a single global maximum, and they allow the two leaders to find an equilibrium position that would not exist in the original structure of the game.

\begin{figure}[h]
\begin{subfigure}[t]{.48\linewidth}
    \includegraphics[width=\linewidth]{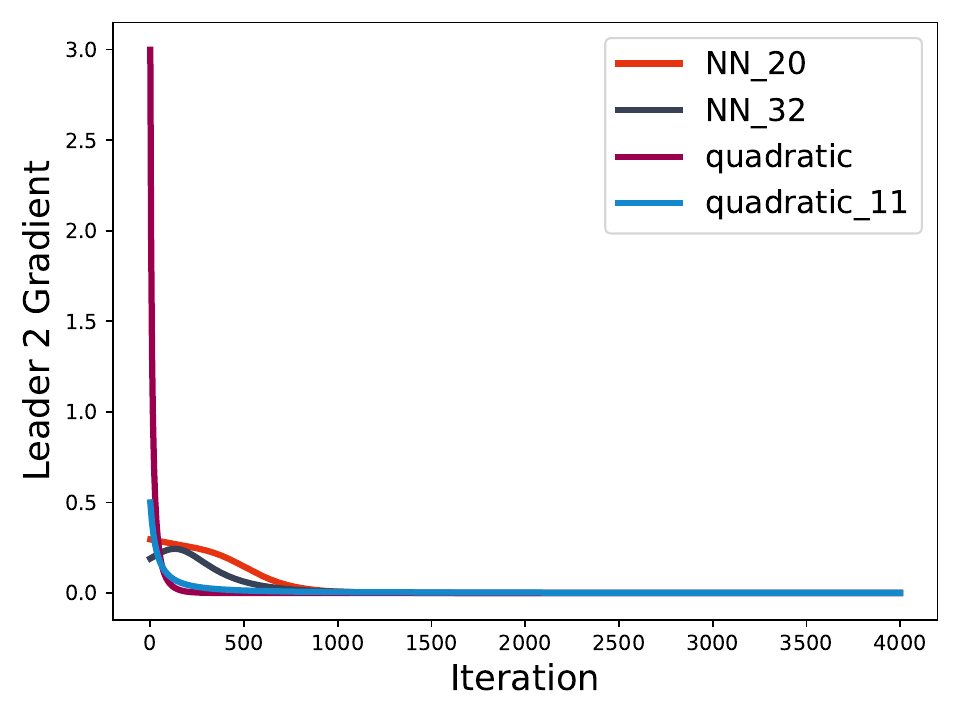}
\caption{Evolution of the gradient of leader 2.}
\label{fig:zonal_x2_grad}
\end{subfigure}
\hfill
\begin{subfigure}[t]{.48\linewidth}
    \includegraphics[width=\linewidth]{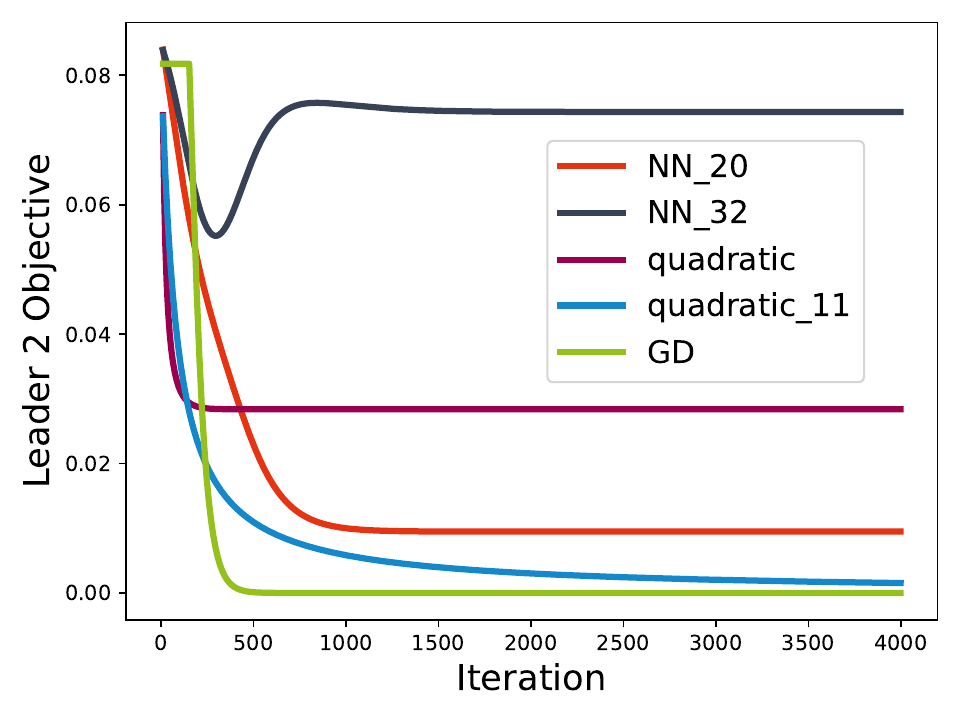}
\caption{Evolution of the objective function of leader 2.}
\label{fig:zonal_x2_obj}
\end{subfigure}
\caption{Gradient and objective evolution of leader 2 in the Leader's Dilemma game.}
\label{fig:zonal_x2}
\end{figure}

\begin{figure}[h]
\begin{subfigure}[t]{.48\linewidth}
    \includegraphics[width=\linewidth]{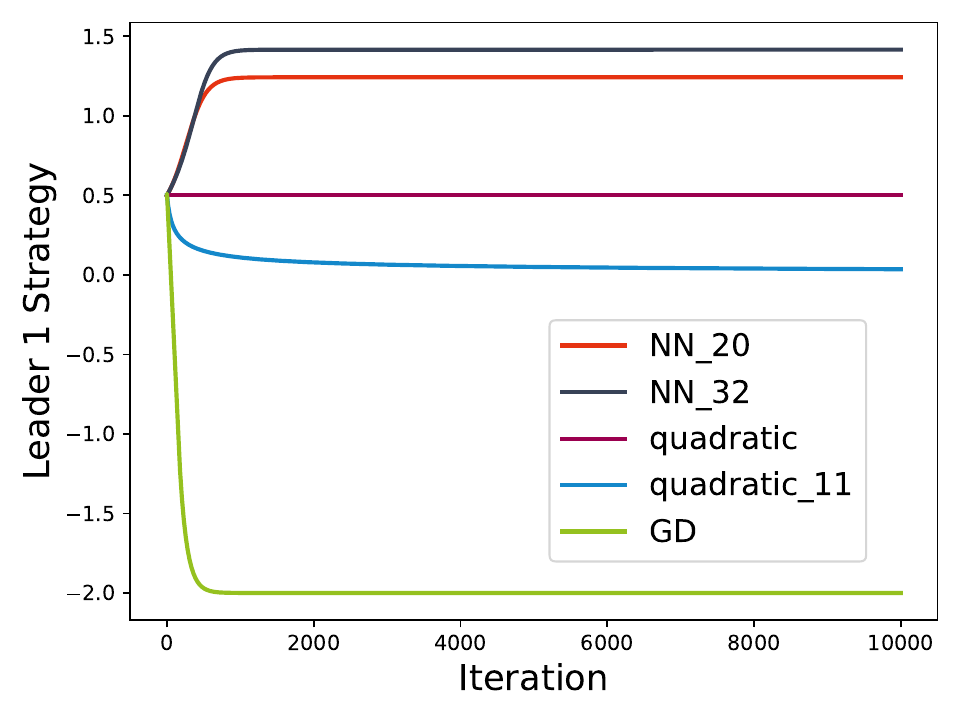}
\caption{Evolution of the strategy of leader 1.}
\label{fig:zonal_x1_evol}
\end{subfigure}
\hfill
\begin{subfigure}[t]{.48\linewidth}
    \includegraphics[width=\linewidth]{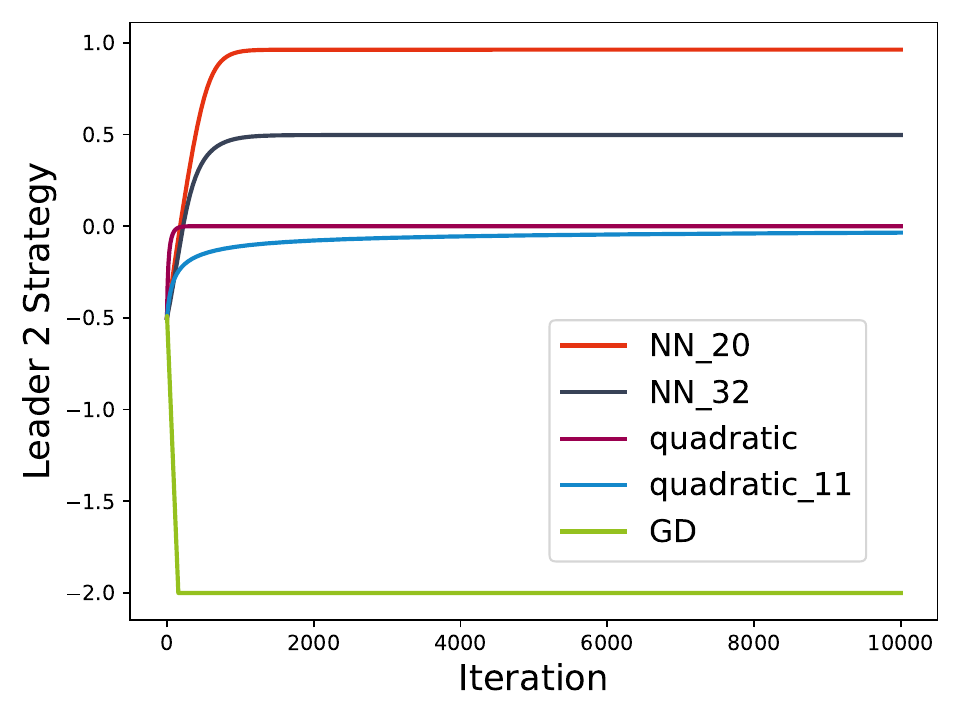}
\caption{Evolution of the strategy of leader 2.}
\label{fig:zonal_x2_evol}
\end{subfigure}
\caption{Plots of the evolution of the strategies of the players of the Leader's Dilemma game.}
\label{fig:zonal_strat}
\end{figure}

\begin{figure}
    \centering
    \includegraphics[width=0.7\linewidth]{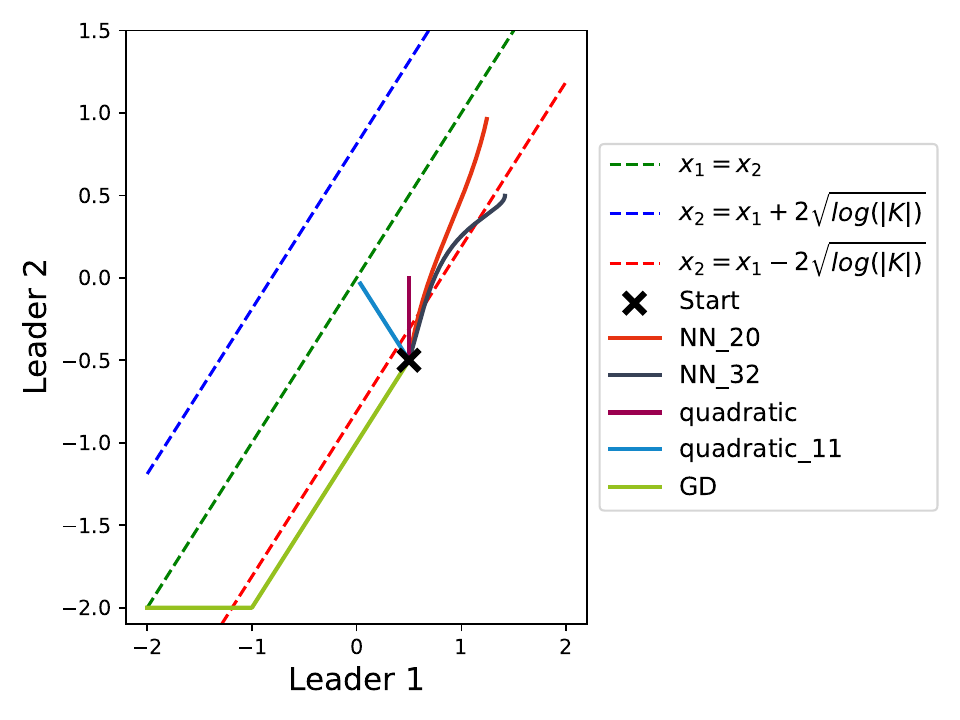}
    \caption{Comparison between the evolution of the strategies of the two leaders and the Stackelberg Equilibria of the Leader's Dilemma game}
    \label{fig:zonal_x1_x2}
\end{figure}

\begin{figure}
    \centering
    \includegraphics[width=0.5\linewidth]{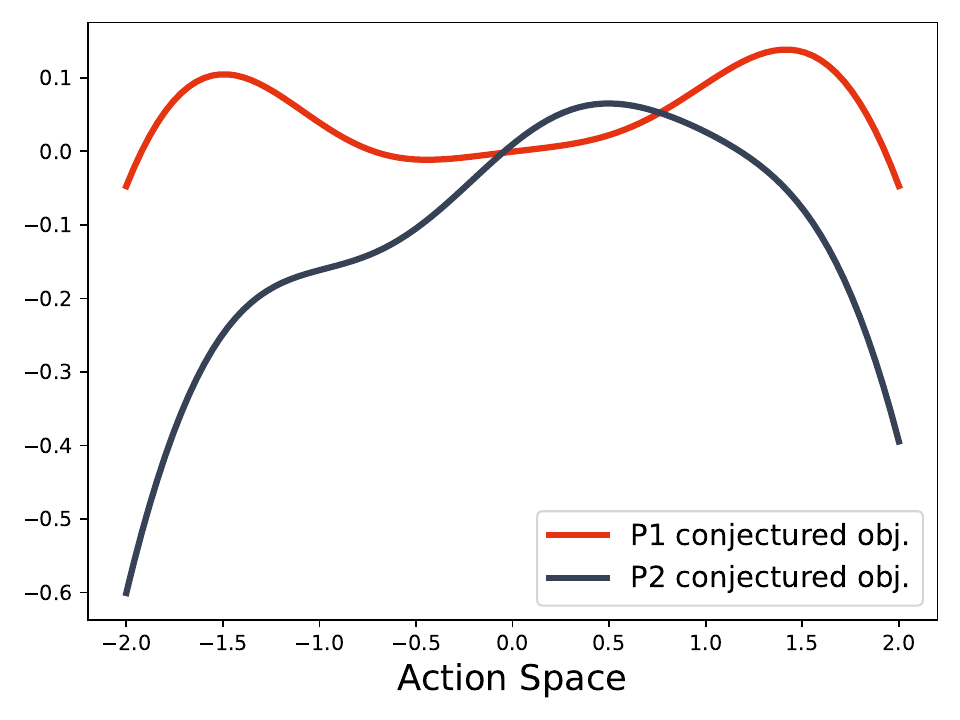}
    \caption{Objective functions of the two leaders evaluated with the learned conjectures \texttt{NN$\_$32}.}
    \label{fig:zonal_conj_obj}
\end{figure}

\paragraph{Revisiting Olsder's Paradox.} We report a similar analysis also for the second game. In Fig.~\ref{fig:olsder_strat} we can see once again how the strategies of the players converge very quickly, in this case also for the \texttt{GD} update. As before we show also the convergence of the gradients in Fig.~\ref{fig:olsder_grad}, where we stress that in the two plots we are only showing the first 500 iterations, to highlight the fact that the gradient does go to 0 extremely quickly in this game. Furthermore, we also report in Fig.~\ref{fig:olsder_obj}, the evolution of the objective functions, comparing them with the value assumed in the different equilibria. Once again we can observe how quickly these functions converge to a stable point, without changing throughout the rest of the simulation.

\begin{figure}[h]
\begin{subfigure}{.48\linewidth}
    \includegraphics[width=\linewidth]{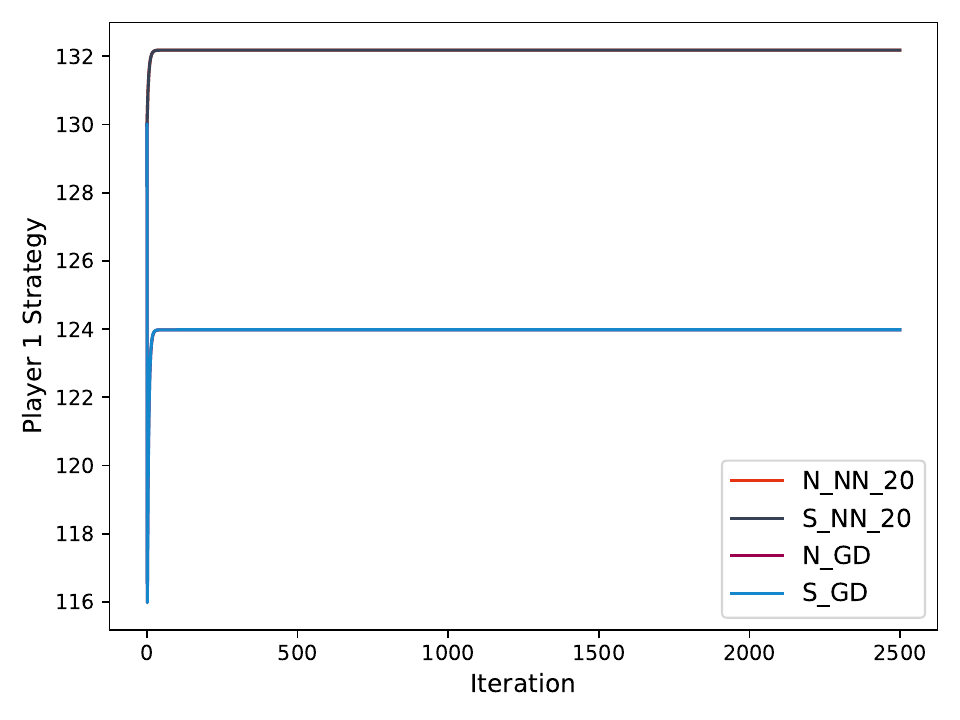}
\caption{Evolution of the strategy of player 1.}
\label{fig:olsder_x1}
\end{subfigure}
\hfill
\begin{subfigure}{.48\linewidth}
    \includegraphics[width=\linewidth]{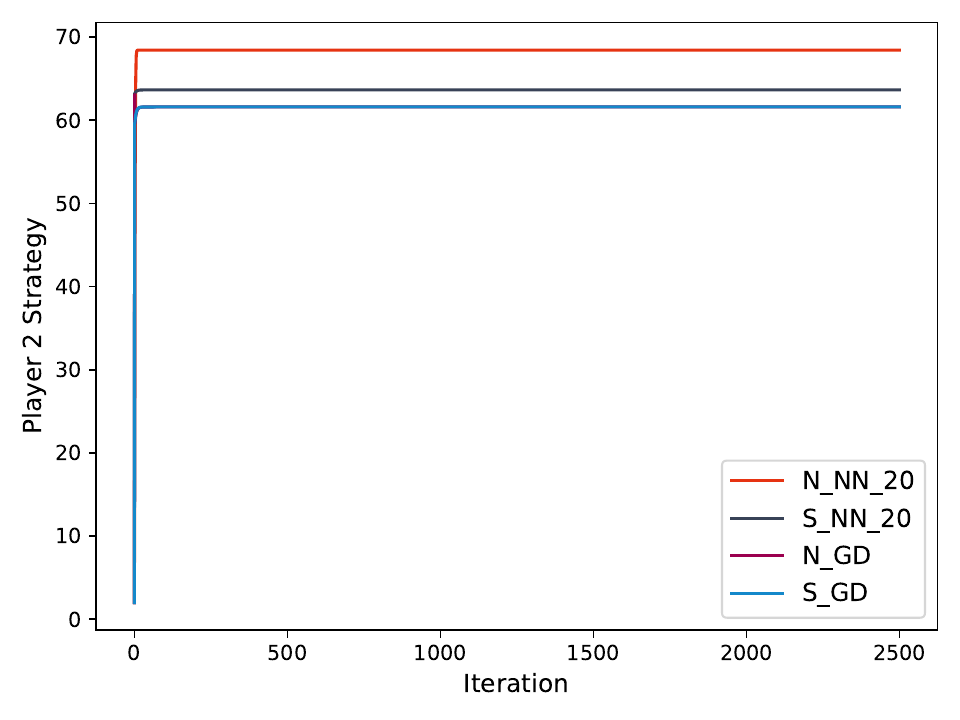}
\caption{Evolution of the strategy of player 2.}
\label{fig:olsder_x2}
\end{subfigure}
\caption{Strategies of the two players throughout the simulation.}
\label{fig:olsder_strat}
\end{figure}

\begin{figure}[h]
\begin{subfigure}{.48\linewidth}
    \includegraphics[width=\linewidth]{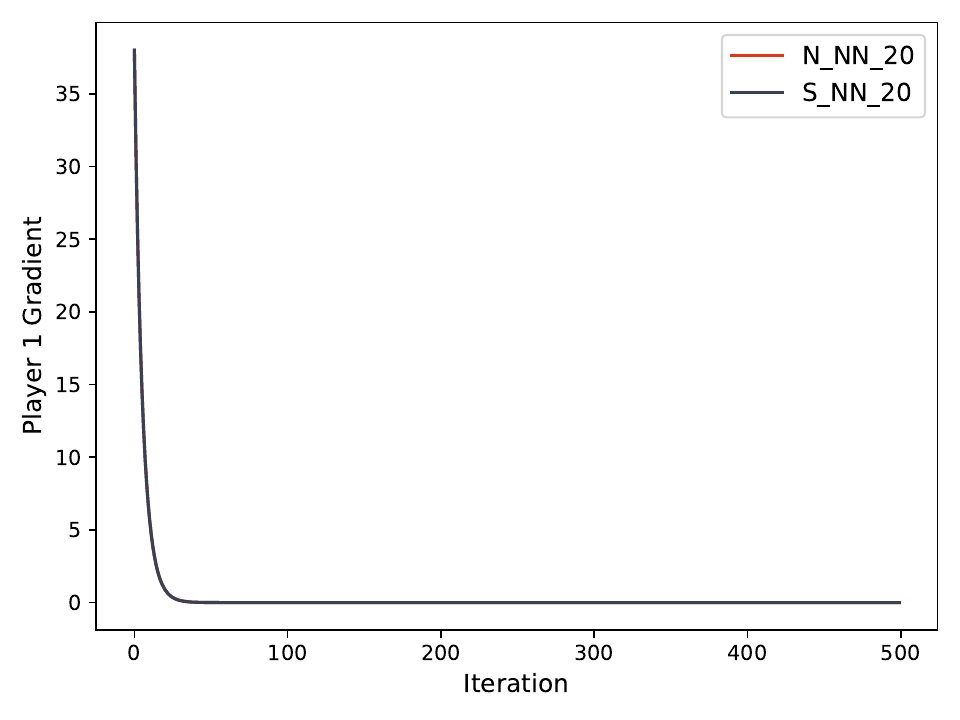}
\caption{Evolution of the gradient of player 1.}
\label{fig:olsder_grad_1}
\end{subfigure}
\hfill
\begin{subfigure}{.48\linewidth}
    \includegraphics[width=\linewidth]{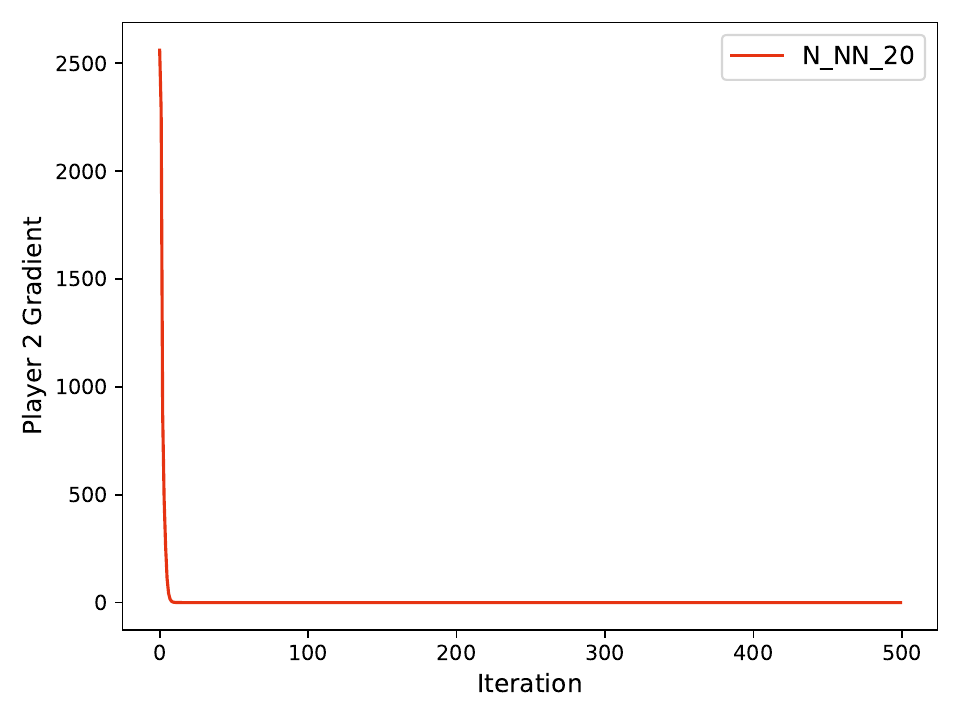}
\caption{Evolution of the gradient of player 2.}
\label{fig:olsder_grad_2}
\end{subfigure}
\caption{Gradients of the two players throughout the simulation.}
\label{fig:olsder_grad}
\end{figure}

\begin{figure}[h]
\begin{subfigure}{.48\linewidth}
    \includegraphics[width=\linewidth]{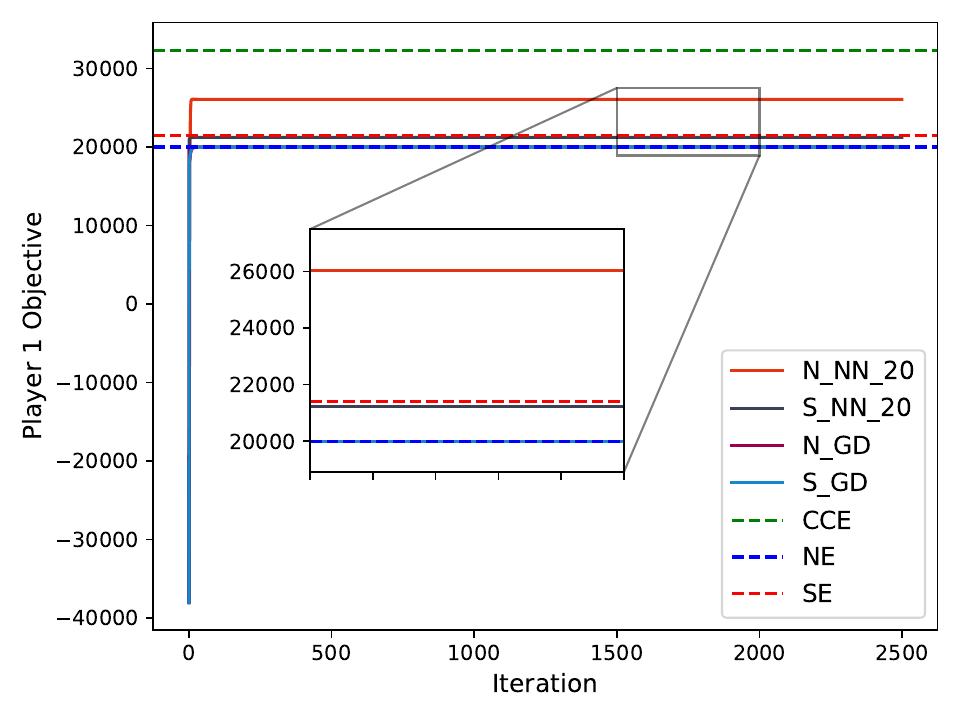}
\caption{Evolution of the objective function of player 1.}
\label{fig:olsder_obj_1}
\end{subfigure}
\hfill
\begin{subfigure}{.48\linewidth}
    \includegraphics[width=\linewidth]{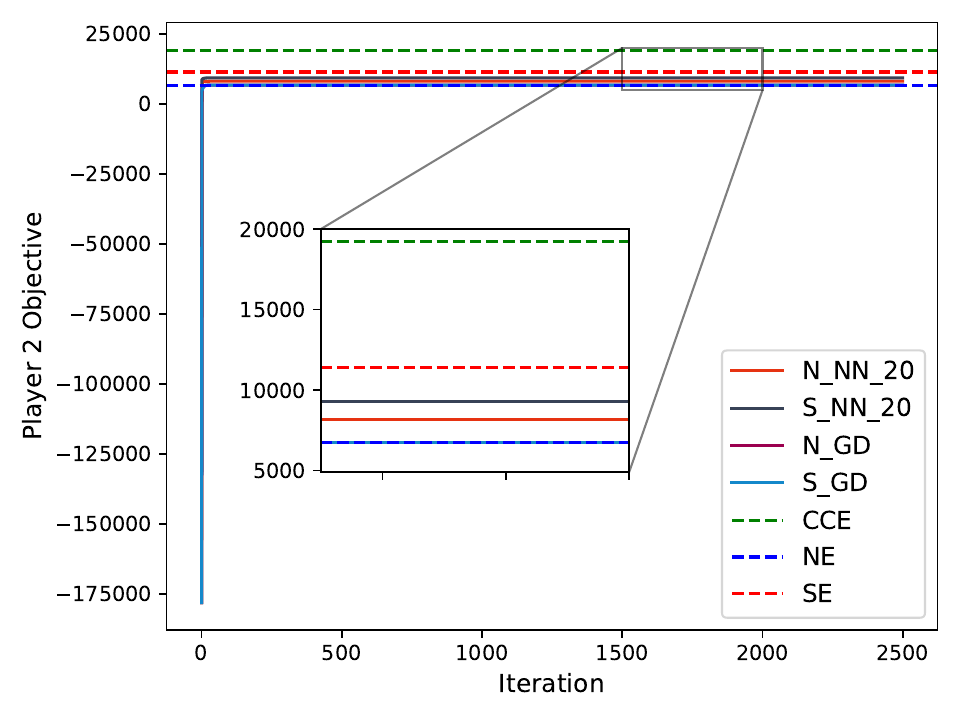}
\caption{Evolution of the objective function of player 2.}
\label{fig:olsder_obj_2}
\end{subfigure}
\caption{Objective functions of the two players throughout the simulation.}
\label{fig:olsder_obj}
\end{figure}

\end{document}